\newif\ifanonymous
\newif\ifdraft

\anonymousfalse
\draftfalse
\InputIfFileExists{localflags}{}
    
\ifdraft \overfullrule=1mm \fi

\documentclass[submission,copyright,creativecommons]{eptcs}
\providecommand{\event}{CREST 2019} 

\usepackage{versions}
\excludeversion{full}
\includeversion{conf}

\usepackage{times}  
\usepackage{helvet}  
\usepackage{courier}  
\usepackage{url}  
\usepackage{graphicx}  
\usepackage{amsthm} 
\usepackage{amsmath}
\usepackage{amssymb}
\usepackage{stmaryrd}
\usepackage{xspace}
\usepackage{proof}
\usepackage{color}
\usepackage{paralist}
\usepackage{booktabs} 
\usepackage{hyperref}

\usepackage{listings} 
\usepackage{syntax} 

\usepackage{footnote}
\makesavenoteenv{tabular}
\makesavenoteenv{table}

\newcommand{\shortcite}[1]{\cite{#1}}

\makeatletter
\DeclareRobustCommand{\defeq}{\mathrel{\rlap{%
  \raisebox{0.3ex}{$\m@th\cdot$}}%
  \raisebox{-0.3ex}{$\m@th\cdot$}}%
  =}
\DeclareRobustCommand{\eqdef}{=\mathrel{\rlap{%
  \raisebox{0.3ex}{$\m@th\cdot$}}%
  \raisebox{-0.3ex}{$\m@th\cdot$}}%
  }
\makeatother

\newcommand{\set}[1]{\{\,#1\,\}}

\newtheorem{definition}{Definition}

\newtheorem{theorem}{Theorem}
\newtheorem{example}{Example}

\newtheoremstyle{mycase}{}{}{}{}{\bf}{.}{.5em}{\thmnote{#1:}~\normalfont
  #3}
\theoremstyle{mycase}

\numberwithin{subcase}{mycase}

\newtheoremstyle{component}{}{}{}{}{\itshape}{.}{.5em}{\thmnote{#3}#1}
\theoremstyle{component}

\newcommand{\ledot}{\mathrel{\ooalign{\hss\raise.200ex\hbox{$\cdot$}\hss\cr$\le$}}}
\newcommand{\gedot}{\mathrel{\ooalign{\hss\raise.200ex\hbox{$\cdot$}\hss\cr$\ge$}}}

\DeclareMathOperator*{\bigtimes}{\vartimes}

\newcommand\bit{\lbrace0,1\rbrace}

\newcommand\setN{\mathbb N}

\newcommand{\calF}{\ensuremath{\mathcal{F}}\xspace}

\newcommand\calR{\ensuremath{\mathcal{R}}\xspace}
\newcommand\calS{\ensuremath{\mathcal{S}}\xspace}

\newcommand{\calU}{\ensuremath{\mathcal{U}}\xspace}
\newcommand\calV{\ensuremath{\mathcal{V}}\xspace}
\newcommand\calX{\ensuremath{\mathcal{X}}\xspace}
\newcommand\calY{\ensuremath{\mathcal{Y}}\xspace}

\newcommand{\mathcmd}[1]{{\normalfont\ensuremath{#1}}\xspace}

\newcommand{\mathfun}[1]{\mathcmd{\mathit{#1}}}

\newcommand{\mathvalue}[1]{\mathcmd{\mathit{#1}}}

\newcommand{\textop}[1]{\relax\ifmmode\mathop{\text{#1}}\else\text{#1}\fi}
 
\DeclareMathOperator\E E

\newcommand{\ctl}{\mathfun{ctl}}
\newcommand{\dat}{\mathfun{dat}}
\newcommand{\rch}{\mathfun{rch}}

\newcommand{\minX}{\mathfun{min}_X}

\newcommand{\Vres}{\calV^\mathit{res}}

\newcommand{\Vpos}{V^\rch_\mathit{pos}}

\newcommand{\Rpoisoned}{\mathvalue{P}_r}
\newcommand{\Rnotpoisoned}{\mathvalue{NP}_r}
\newcommand{\Rneutralized}{\mathvalue{N}_r}

\newcommand{\Vrch}{\calV_\rch}
\newcommand{\Vdat}{\calV_\dat}
\newcommand{\Gctl}{G_\ctl}

\newcommand{\sM}{(M,\vec u)\vDash}

\newcommand{\BT}{\mathit{BT}}
\newcommand{\ST}{\mathit{ST}}

\newcommand{\BS}{\mathit{BS}}

\newcommand{\exnumber}{34}

\begin{document}

\title{Causality \& Control Flow}
\def\titlerunning{Causality \& Control Flow}
\ifanonymous 
\def\authorrunning{Anonymous}
\author{}
\else 
\author{
    Robert K\"{u}nnemann
    \institute{ CISPA Helmholtz Center for Information Security}
    \and 
    Deepak Garg
    \institute{
        MPI-SWS}
    \and 
    Michael Backes
    \institute{ CISPA Helmholtz Center for Information Security}
} 
\def\authorrunning{K\"{u}nnemann, Garg \& Backes}
\fi
\begin{full}
     \date{}
\end{full}
\maketitle

\begin{abstract}
Causality has been the issue of philosophic debate since
Hippocrates. It is used in formal verification and testing, e.g.,
to explain counterexamples or construct fault trees.  Recent work
defines actual causation in terms of Pearl's causality framework,
but most definitions brought forward so far struggle with examples
where one event preempts another one.  A key point to capturing
such examples in the context of programs or distributed systems is
a sound treatment of control flow.  We discuss how causal models
should incorporate control flow and discover that much of what
Pearl/Halpern's notion of contingencies tries to capture is
captured better by an explicit modelling of the control flow in
terms of structural equations and an arguably simpler definition.
Inspired by causality notions in the security domain, we bring
forward a definition of causality that takes these
control-variables into account.  This definition provides a clear
picture of the interaction between control flow and causality and
captures these notoriously difficult preemption examples without
secondary concepts.  We give convincing results on a benchmark of
\exnumber{} examples from the literature.
\end{abstract}

\section{Introduction}

A growing body of literature is concerned with notions of
accountability in security protocols, as in many scenarios, e.g.,
electronic voting, certified e-mail, online transactions, or when
personal data is processed within a company, not all agents can be
trusted to behave according to some established protocol~\cite{Bella:2006:APF:1151414.1151416}. 
Thus, in order to
specify accountability, i.e., the
ability of a protocol to detect misbehaviour, 
\begin{full}
but also to reason about
test coverage~\cite{Chockler:2008:CSS:1352582.1352588},
or explain attacks~\cite{Beer2009},
\end{full}
the information security domain needs a reliable notion of
what it means for
a protocol event to \emph{cause} a security violation
in a given scenario.
The security domain has proposed causal notions specifically for
network traces, which, in contrast to traditional notions of
causality, capture actions sufficient to cause an event 
and put a focus on control flow~\cite{Datta2015a}.
In this work, we investigate these ideas in a more general setting;
generalising, improving and validating them to provide a sound basis
for causal reasoning in protocols and beyond.

The problem we investigate is called \emph{actual causation},
as opposed to \emph{type causation}, which aims at deriving 
general statements, e.g., ``smoking causes cancer'' not
linked to a specific scenario.
Starting from 
Lewis' `closest-world concept'~\shortcite{Lewis1973-LEWC},
philosophers have largely accepted
\emph{counter-factual} reasoning, i.e.,
investigating causal claims by
regarding hypothetical scenarios
of the form `had A not occurred, B would not have occurred',
as a means to determine actual causation.
Pearl's causality framework~\shortcite{pearl-book}
provides a basis for such reasoning. 

So far, control flow has largely been ignored in causal reasoning.
This is not very surprising, considering that the causation literature
typically treats real-life examples inspired from criminal
law. Control flow is simply not a well-defined notion there.
Albeit, precisely those scenarios where the order of
events is relevant, e.g., 
an event might
prevent another event from happening, 
turn out
to be notoriously difficult to explain using counterfactual reasoning.
Once we consider each potential course of events as
a control-flow path consisting of events that may enable or prevent
each other, they become easy to handle.

To accommodate preemption, 
Pearl and Halpern's very influential notion of causation
has been modified several
times~\cite{DBLP:conf/ijcai/Halpern15},
complemented with secondary notions~\cite{DBLP:journals/corr/HalpernH13}
and
ad-hoc modifications have been proposed~\cite[p.~26]{DBLP:journals/corr/HalpernH13}.
Neither of these solutions provides a satisfying answer as to how these
examples should be handled in general.

In this work, we provide an account of the relation between control
flow and causation. This gives us the means to adequately capture
preemption in cases where we can speak of control flow.  
We propose that control flow variables should be modelled
explicitly, as they can capture the course of events that lead to a certain
outcome. Once they are made explicit,
we can capture these difficult examples and provide
a notion of actual causation that 
is simple and intuitive,
captures joint as well as independent causes,
gets by without secondary notions of normality and defaults
and
readily applies to Pearl's causality framework.
Our contributions are the following:
\begin{enumerate}

    \item We explain how control flow should be incorporated in causal
        models 
        and propose a formalism that
        makes control flow explicit. 

    \item We show how control flow helps to handle
        preemption
        without resorting to
        secondary
        notions like defaults and normality or ad-hoc modifications to the
        model. 
        Preemption examples are
        notorious for being difficult to
        handle~\cite{causal-powers,Hitchcock2007-HITPPA,DBLP:journals/corr/HalpernH13},

    \item We relate control flow to 
        \emph{structural contingencies} 
        introduced by Halpern and Pearl~\cite{DBLP:journals/corr/abs-1301-2275},
        providing evidence that
        what this notion achieves is very similar to a simple fixing of
        control flow variables, and that it provides unintuitive results when
        applied outside control flow.

    \item Finally, we validate our proposal, control flow preserving
        sufficient causation, on $\exnumber{}$ examples,
        including all scenarios discussed in~\cite{weslake2015partial}
        and~\cite{DBLP:conf/ijcai/Halpern15}.

\end{enumerate}

\paragraph{Notation}\label{sec:notation}%
%
%
%
%
%
%
We write $\vec t$ for a sequence $t_1,\dotsc,t_n$ if $n$ is clear from
the context and 
use
$(a_1,\ldots,a_n) \cdot (b_1,\ldots,b_m)=
(a_1,\ldots,a_n, b_1,\ldots,b_m)$
to denote concatenation.
%
%
%
%
We filter a sequence $l$ by a set $S$, denoted $l|_S$, by removing each element
that is not in $S$. 
%

\section{Causality framework (Review)}

We review the causality framework introduced by
Pearl~\shortcite{pearl-book},
also known as the \emph{structural equations model}.
The causality framework models how random variables influence
each other. The set of random variables, which we assume discrete,
is partitioned into
a set $\calU$ of \emph{exogenous} variables, variables that are
outside the model, e.g., in the case of a security protocol, the
scheduling and the
attack the adversary decides to mount, and
a set $\calV$ of \emph{endogenous} variables, 
which are ultimately determined by the value of the
exogenous variables. 
A \emph{signature} is
a triple consisting of $\calU$, $\calV$ and function $\calR$
associating a range, i.e., a set, to each variable $Y\in\calU \cup \calV$.
A causal model on this signature defines the relation between
endogenous variables and exogenous variables or other endogenous
variables in terms of a set of equations.
\begin{definition}[Causal model]\label{def:causal-model} 
    A \emph{causal model} $M$ over a signature $\calS=(\calU,\calV,\calR)$
    is a pair of said signature $\calS$ and a set of functions
    $\calF=\set{F_X}_{X\in\calV}$ such that, for each $X\in\calV$,
    \[ F_X : (\bigtimes_{U\in\calU} \calR(U) ) \times
    (\bigtimes_{Y\in\calV\setminus \set{X}} \calR(Y) ) 
\to \calR(X). \]
\end{definition}

Each causal model induces a \emph{causal
network}, a graph with
a node for each variable in
$\calV$, and an edge from $X$ to $Y$ iff $F_Y$
depends on $X$.
($Y$ depends on $X$ iff there is a setting for the variables in 
$\calV\cup\calU\setminus\set{X,Y}$
such that modifying $X$ changes the value of $Y$.)
If the causal graph associated to a causal model $M$ is
acyclic, then each
setting $\vec u$ of the variables in
$\calU$ provides a unique solution to the equations in $M$. 
Throughout this paper, we only consider 
causal models that have this property.
We call a vector setting the variables in $\calU$
a \emph{context}, and 
a pair $(M,\vec u)$ of a causal model and a context \emph{a
situation}.
All modern definitions of
causality follow 
a counterfactual approach, 
which requires answering `what if'
questions. 
\begin{full}
It is not always possible to do this by observing actual outcomes.
Consider the following example.

\begin{example}[Wet ground]\label{ex:wet-ground}
The ground in front of Charlie's house is slippery when wet. Not only
does it become wet when it rains; if the neighbour's sprinkler is
turned on, the ground gets wet, too. The neighbour turns on the
    sprinkler unless it rains. Let $R\in\calU$ be 1 if it rains, and
    0 otherwise, and consider the following equations for a causal
      model on $R$ and endogenous variables $W$, $S$ and $F$ with
      range $\bit$.
      \begin{align*}
          S & = \neg R & \text{(The sprinkler is on if it does not
          rain.)}\\
          W & = R \lor S & 
          \text{(The sprinkler or the rain wets the ground.)}\\
          F & = W & \text{(Charlie falls when the ground is
          slippery.)}
      \end{align*}
\end{example}
Clearly, the ground being wet is a cause to Charlie's falling,
but we cannot argue counterfactually, because the counterfactual case
never actually occurs: the ground is always wet. We need to intervene
on the causal model.
\end{full}

\begin{definition}[Modified causal model]\label{def:modified-causal-model} 
    Given a causal model 
    $M=((\calU,\calV,\calR),\calF)$, we define the 
    \emph{modified causal model}
    $M_{\vec X \leftarrow \vec x}$ over the signature
    $\calS_{\vec X}=(\calU,\calV \setminus \vec X, 
    \calR|_{\calU \cup \calV \setminus \vec{X}})$
    by replacing each endogenous variable
    $X\in\vec X$ 
    in $\calF$
    with 
    the corresponding $x\in\vec x$,
    obtaining $\calF_{\vec X \leftarrow \vec x}$. Then,
    $M_{\vec X \leftarrow \vec x} = (\calS_{\vec X}, \calF_{\vec X \leftarrow \vec x})$.
\end{definition}

We can now define how to evaluate queries on
causal models w.r.t.\ interventions on a vector of variables,
which allows us to answer `what if' questions.

\begin{definition}[Causal formula]\label{def:basic-causal-formula} 
    A \emph{causal formula} has the form
    $[Y_1 \leftarrow y_1, \ldots, Y_n \leftarrow y_n]\varphi$
    (abbreviated $[\vec Y \leftarrow \vec y]\varphi$), where
    \begin{itemize}
        \item $\varphi$ is a boolean combination of primitive events,
            i.e., formulas of the form $X=x$ for $X\in\calV$,
            $x\in\calR(X)$,
        \item $Y_1,\ldots,Y_n\in\calV$ are distinct,
        \item $y_i\in\calR(Y_i)$.
    \end{itemize}
    We write $(M,\vec u)\vDash [\vec Y \leftarrow \vec y]\varphi$
    \begin{full}
    ($[\vec Y\leftarrow \vec y]\varphi$ 
    is true in a causal model $M$ given
    context $\vec u$) %
    \end{full}
    if the (unique) solution to the
    equations in $M_{\vec Y\leftarrow \vec y}$ in the context $\vec u$ 
    satisfies $\varphi$.
\end{definition}


\section{Sufficient causes}\label{sec:sufficient}

We define a causality notion based on sufficiency.
In order to
allow for comparison with existing notions of causation,
we chose to
formulate this notion in Pearl's causation framework,
as opposed to formalisms that already incorporate temporality, e.g.,
Kripke structures, but would obscure this comparison.
This simplistic notion of causality, which we will later extend to
models with explicit control flow, captures the causal variables that,
by themselves, guarantee the outcome. 


\begin{definition}[Sufficient cause]\label{def:sufficient-cause} 
    $\vec X = \vec x$ is a \emph{sufficient cause} of $\varphi$ in
    $(M,\vec u)$ if the following three conditions hold.
    \begin{enumerate}
        \item [SF1.] $(M,\vec u)\vDash (\vec X = \vec x) \land
            \varphi$.
        \item [SF2.] For all $\vec z$,
            $(M,\vec u) \vDash [(\calV \setminus \vec X)\leftarrow \vec z] \varphi$.
        \item [SF3.] $\vec X$ is minimal: No strict subset 
            $\vec X'$ of $\vec X$ satisfies SF1 and SF2.
    \end{enumerate}
    We say $\vec X$ is a sufficient cause for $\varphi$ if
    this is the case for some $\vec x$. Any non-empty subset of
    $\vec X$ is \emph{part} of the sufficient cause $\vec X$.
\end{definition}

Sufficient causes are well-suited for establishing 
joint causation, i.e.,
several factors that independently would not cause an injury, but do
so in combination.
Consider the following scenario: 
\begin{example}[Forest fire, conjunctive]\label{ex:forest-conjunctive}
Person $A$ drops a canister full of gasoline in the forest, which
soaks a tree. An hour later, $B$ smokes a cigarette next to that
tree. A and B have joint responsibility for the resulting forest fire
($\mathit{FF} = A \land B$). The above definition yields $(A, B, FF) = (1,1,1)$ as
the sufficient cause.
\end{example}

In contrast,  
the traditional but-for test
(or \emph{condicio sine qua non}),
formulates a necessary condition. 
$A$ and $B$ on their own are necessary causes, despite the fact that
the fire was jointly caused.
Sufficient causation can
distinguish joint causation, like in this case,
from independent causation, like in the case where the forest was dry and both $A$
and $B$ dropped cigarettes independently.

\emph{Remark:} %
Similar to how 
the set of necessary
causes for any $O=o$ contains
the singleton cause $O=o$,
$O=o$ is also
a part of each sufficient cause for $O=o$.
Hence, it can be filtered out for most purposes.


\section{Modelling control flow}\label{sec:controlflow}

In this section, we discuss how control flow should be incorporated
into causal models.
Consider the famous \emph{late preemption} example,
where Suzy and Billy both throw stones at
a bottle.
Suzy's stone hits the bottle first and shatters it, while
Billy would have hit, had the bottle still been 
there~\cite[Example~3.2]{DBLP:conf/ijcai/Halpern15}.
We will discuss
several models of this example.

\begin{example}[Control flow in equations]\label{ex:ctl-eq}
    Exogenous variables $\ST$ and $\BT$ are 1 if Suzy, respectively
    Billy, throws. The endogenous variable $\BS$ is 1 (bottle shatters) if $\ST\lor\BT$.
\end{example}
Pearl and Halpern propose a slightly modified model of this situation
which captures the relationship between the bottle being shattered and the
bottle being hittable by Billy explicitly~\shortcite{DBLP:journals/corr/abs-1301-2275}.
However, the relationship
between the two is fixed in the model. 
Datta, Garg, Kaynar and Sharma observe this and therefore 
introduce exogenous variables to
determine which stone reaches the bottle first depending on the
context~\shortcite{tracing-causes}.
\begin{example}[Control flow in context]\label{ex:ctl-context}
    Exogenous variables $\ST$ and $\BT$ are 1 if Suzy, respectively
    Billy throws, and $R$ is 1 if Suzy's throw reaches the bottle
    first. Then $\BS\in\calV$ is $\ST$ if $R=1$, and otherwise $\BT$.
\end{example}
Pearl and Halpern's solution can be transferred to the case where the
order is not fixed a priori by
explicitly representing the temporal order in
distinguished variables.
\begin{example}[Control flow in variables]\label{ex:ctl-var}
    Exogenous variables $T_1$ to $T_n$ with range $\set{S,B,N}$ model
    whether Suzy, Billy or no-one throws a stone at point $i$. 
    Endogenous variable $\BS_i$ is 1 if $\BS_{i-1}=0$ and 
    $T_i\neq N$. $\BS=1$ if $\BS_i=1$ for any $i$.
\end{example}
$BS_i$ models whether the bottle is
available for hitting at point $i+1$,
similar to the concept of
control flow variables in programming.
In programming, control flow is
the order in which statements are evaluated. 
Interpreting the above
model as a program, $\BS_i$ controls whether $\BS$ is assigned 
1, similar to $n$ nested if-statements surrounding an assignment
$\BS\defeq 1$.

\paragraph{Control flow should be modelled within variables.}

Among these three solutions
(Examples~\ref{ex:ctl-eq},~\ref{ex:ctl-context} and~\ref{ex:ctl-var}),
only the second and third solutions are able to capture preemption,
as preemption is about the temporal relation between
events, and control flow captures just that. Without examining the
equations themselves, it is not possible to distinguish this order, as
Example~\ref{ex:ctl-eq} demonstrates.
The second solution is not always sufficient, 
since control flow is often determined by data flow,
e.g., when branching on a variable.
Hence, preemption can be accounted for by modelling control flow and
data flow separately, at least in scenarios where these concepts are
meaningful, e.g., in programs.
Even for distributed systems, once the scheduler is made explicit, the
system can be adequately modelled as a program and hence control flow
can be distinguished.
In the following, we will demonstrate this point on a number of other
examples (which admittedly have little to do with programs, but we
want to stay close to the literature. They can be recast easily,
imagine, e.g., Suzy's and Billy's stone being illegal instructions in
a message queue).
Note that there might be ways of modelling preemption which do not
fall in any of these three classes, however, all examples we are aware
of follow one of these three paradigms.
We formalise our assumptions on control flow as follows.

\begin{definition}[Causal model with control flow]\label{def:ext-causal-model} 
    We say an
    acyclic causal model $M=(S,\calF)$ over a signature
    $\calS=(\calU,\calV,\calR)$
    has control flow variables $\Vrch$
    if
    \begin{enumerate}
        \item $\calV$ can be partitioned into $\Vrch$ and a set of
            data variables $\Vdat$,

        \item $\calR(V^\rch)=\set{\top,\bot}$ for any
            $V^\rch\in\Vrch$, and,

        \item for $\Gctl$ the subgraph of $M$'s causal network
            containing all nodes in $\Vrch$ and all edges between
            them, and for all contexts $\vec u$ and 
            data variable assignments $\vec v$, 
            if all parents $\vec V^\rch_p$ of a node $V^\rch\in \Gctl$ are set to
            $\bot$, this node is also set to $\bot$, i.e.,
            if $\sM [\vec \Vdat \gets \vec v] (V^\rch_p =\bot)$ for each
            parent $\vec V^\rch_p$, then
            $\sM [\vec \Vdat \gets \vec v] (V^\rch=\bot)$. \label{it:valid}


    \end{enumerate}
\end{definition}

The active variables in $\Gctl$, i.e., those equal to $\top$,
represent the current control flow.
Data flow variables are assigned depending on which nodes in $\Gctl$
are active.

For instance:
if we consider structured control flow,
$\Gctl$ is connected and thus
(because $M$ is acyclic) a tree.
Furthermore, the set $\set{V^\rch \mid (M,u)\vDash V^\rch=\top}$ always comprises a path.
The relationship to non-concurrent imperative programming languages
becomes clearer, if we express each function $F_V$ for a data
variable $V\in\Vdat$ as 
\[ F_V = \begin{cases}
    F_{V,V^\rch_1} & \text{if $\Vpos=V^\rch_1$} \\
    & \vdots \\
    F_{V,V^\rch_n} & \text{if $\Vpos = V^\rch_n$} \\
    F_{V,V^\rch_p} & \text{otherwise}
\end{cases}
\]
where $\Vpos$
is
the current position, i.e.,
the deepest element of the path $\set{V^\rch \mid (M,u)\vDash V^\rch=\top}$,
$V^\rch_p$ the parent of $\Vpos$,
and
$F_{V,V^\rch_1},\ldots, F_{V,V^\rch_n}$
depend only on variables in $\Vdat$.
Essentially, at each control-flow position,
$F_{V,V^\rch_i}$ may
assign a value to $V$, or otherwise
the parent assignment remains in
effect.

\paragraph{Intervention on control flow variables.}

For most (but not all) purposes, we are not interested in control flow
variables as parts of sufficient causes.
We could filter them out, i.e., if $\vec X$ is
a sufficient cause, report $\vec X \setminus \Vrch$.
%
%
Instead, we chose to achieve this by
restricting intervention to a subset of variables 
$\Vres = \Vdat = \calV\setminus \calV_\rch$.
Both approaches model different expectations on the system. We want to
assume the control flow to be (locally) determined and not
consider failure in, e.g., conditional branching.
The same assumption is made by
Datta,
Garg,
Kaynar,
Sharma and
Sinha
(DGKSS)~\shortcite{Datta2015a} and
Beckers~\shortcite{545701}.

\begin{full}
Upon close inspection, this is similar to what DGKSS's definition of
causal traces achieves. The process calculus they propose 
implicitly restricts intervention to
the data flow.
On the other hand, it does not support branching,
so it is difficult to assess the implications within this calculus.
This models different expectations on the system: do we want to
assume the control flow to be (locally) determined, or are we
considering failure in, e.g., conditional branching.
If we never assume failure, we can gather fewer causes, but all of
them correspond to a sufficient cause without this assumption.

Altering condition \emph{SF2} in
Definition~\ref{def:sufficient-cause}, we can model restricted
intervention as
\begin{itemize}
\item [SF2'] For all $\vec z$, $\sM [\Vres \setminus \vec X \gets z] \varphi$,
    with $\Vres=\calV\setminus \calV_\rch$.
\end{itemize}

\begin{theorem}
All (minimal) sufficient causes $\vec X$ with restriction subsume a
(minimal) sufficient cause without restriction. 
\end{theorem}
\begin{proof}
\emph{SF2'}
is equivalent to
\begin{align*}
    & \sM [ \calV \setminus \set{\vec X \cup \calV_\rch} \gets z ] \varphi \forall
    z 
    \\
    \iff & \sM [\calV \setminus \set{\vec X \cup \calV_\rch^{\minX}} \gets z ] \varphi \forall
    z 
\end{align*}
for some minimal set of control flow nodes $\calV_\rch^{\minX}$. 
Given
that $\vec X$ is minimal w.r.t.\ \emph{SF2'}, and hence $\vec X \cap
\calV_\rch = \emptyset$, this sufficient cause
$\vec X \cup \calV_\rch^{\minX}$ is also minimal
w.r.t.\ \emph{SF2}, since for any subset 
$\vec X' \subsetneq \vec X \cup \calV_\rch^{\minX}$,
s.t. $\sM [\calV \setminus X'\gets z]\varphi \forall z$,
\begin{align*}
    & \sM [ \calV \setminus \set{ \vec X'|_{\Vres} \cup
    V_\rch}\gets z]\varphi \;\forall z
    \\
    \iff & 
    \sM [ \Vres \setminus  \vec X'|_{\Vres} \gets z]\varphi \;\forall z.
\end{align*}
Since, $\vec X'|_{\Vres} \subsetneq \vec X$ by minimality of $\calV_\rch^{\minX}$,
this contradicts the minimality of $\vec X$ w.r.t.\ to \emph{SF2'}.
\end{proof}

The converse does not hold, consider, e.g., Halpern's extension of the
conjunctive forest fire example~\cite[Example~3.6]{DBLP:conf/ijcai/Halpern15},
where three different mechanism can produce the fire,
depending on whether
$A\land B$, $\neg A \land B$ or $A \land \neg B$ is the case.
In this case, the possibility that all of these mechanisms fail,
even if $A\lor B$, leads to a third cause containing both $A$ and $B$
(see Table~\ref{tab:examples}, Forest fire, disjunctive, ext.).


%


In summary, forbidding intervention on control flow variables may lose
some granularity in the distinction of causes, but is often appropriate,
when the control flow is assumed to be infallible.
Notice however, that 
none of these notions captures the fact that 
what ever mechanism produces the fire in case $\neg A \land B$ or
$A \land \neg B$ was
not related to its actual coming about.
%
This will be the subject of the next two sections.
\end{full}

\paragraph{What is control flow outside computer programs?}

For a computer program written in an imperative language,
control flow is widely understood to be the order in which statements
are executed, e.g., a sequence of line numbers.
This definition can be easily extended to distributed systems,
however, many examples in the causality literature discuss human
agents in physical interaction.
Our main objective is a reliable notion of causality for distributed
systems, but we also want it  to be grounded in the existing body of
work on causality.
To this end, we make the modelling principles we adhered to explicit.
We neither claim that these modelling
principles are universal, nor that control-flow is a notion that can
be defined in all scenarios where causality applies.
They applied, however, to the \exnumber{} examples we found in the
literature, as we will see.

We assume the modeller has an intuition of how a variable
assignment translates to an (intuitive) `course of events', and when
a  `course of events' should be considered equivalent to
another. 
As discussed 
in the previous paragraph,
we restrict intervention to
data flow
variables.
Hence, we consider
only variable 
assignments 
$\vec \calV = \vec v$, 
that
result from an intervention on the
data flow variables, but not control flow variables, i.e.,
there is a context $\vec u$ and data variable assignment $\vec v_\dat$
s.t.
$\sM [\vec \Vdat \gets \vec v_\dat] (\vec \calV = \vec v)$.
For brevity, we call these assignments \emph{valid}.

\begin{enumerate}
    \item Each control flow variable $V^\rch$ should correspond to
        a relevant event, and vice versa.
        \label{it:correspond}
    \item For every valid assignment, 
        $V^\rch$ should be $\top$ if the
        corresponding event occurred.
        \label{it:top}
    \item For every valid assignment, $V^\rch$ should be $\bot$ if the
        corresponding event did not, or \emph{did not yet},
        occur.
        \label{it:bot}
    \item 
        A control flow variable should only be a parent of another
        control flow variable (in $\Gctl$) iff 
        the occurrence of the
        event corresponding to the child depends on the occurrence of
        the event corresponding to the parent.
        \label{it:parent}
\end{enumerate}

\begin{example}[Early preemption]\label{ex:shot-poisoned}
    Victoria's coffee is poisoned by her bodyguard ($B=1$),
    but before the poison takes effect, she is shot by an assassin
    ($A=1$). She dies ($D=1)$.
\end{example}

Following modelling principles~\ref{it:correspond} and~\ref{it:top},
we introduce at least the control flow variables $P_r$, $S_r$ and $\mathit{PE}_r$,
which, if set to $\top$, represent the events `Victoria is poisoned', `Victoria is
shot' and `the poison takes effect', respectively. 
The poison can only take effect if Victoria was not shot,
but, modelling principle~\ref{it:bot} dictates that $S_r=\bot$ 
only means that she was not \emph{yet} shot, i.e., it might just be
that not enough time has passed, but she will eventually get shot
before the poison takes effect.
We thus need to introduce a fourth control-event, 
`Victoria was not shot during the time the poison needs to take effect',
represented by $\mathit{NS}_r$.
While $F_{S_r} = (A=1)$ and 
$F_{\mathit{NS}_r} = \neg (A=1)$,
and 
thus all valid assignments result in one being the negation of the
other,
we will later consider the coming about of
a course of events, which includes counterfactuals scenarios where
a course of events can be incomplete.
Thus, it is possible that, in a counterfactual scenario,
$S_r=\mathit{NS}_r=\bot$, which intuitively means that not
enough time has passed for the poison to take effect.

Following modelling principle~\ref{it:parent},
$\mathit{PE}_r$ is a child of both $P_r$ and $\mathit{NS}_r$.
The poison takes only effect ($\mathit{PE}_r=\top$) if 
it was administered ($P_r=\top$) and
some time has passed without Victoria getting shot 
($\mathit{NS}_r=\top$).
Following the same principle, $\mathit{PE}_r$ is \emph{not} a child of
$S_r$,
as $\mathit{S}_r=\bot$ could mean
that Victoria is not getting shot at all, but also that she is
\emph{not yet} getting shot. Hence,
$F_{\mathit{PE_r}}= \neg S_r \land P_r$ would be incorrect,
as the poison may or may not take effect due to the shot occurring
later.
Note that the control-flow graph is not linear in this case,
representing the independence of poisoning and the shooting.

\section{Preserving control flow}\label{sec:bogus}
We present further examples that are problematic for
existing definitions of cause in literature, followed by a new
definition of cause (using control flow) that handles all these
examples and many others.
Consider the following example known as
`bogus prevention'~\cite{causal-powers,Hitchcock2007-HITPPA}.

\begin{example}[Bogus prevention, branching]\label{ex:bogus}
An
assassin has a change of heart and refrains from putting poison into
    the victim's coffee ($P=0$). Later, the bodyguard puts antidote into the
    coffee ($A=1$). 
    Is putting the antidote into the coffee the reason the victim
    survives ($S=1$)? 
Let 
    $\calU=\set{U_P,U_A}$,
    $\calV=\set{P,A,S}\cup\set{\Rpoisoned,\Rnotpoisoned,\Rneutralized}$
    and the following equations
    describe a causal model $M_\mathit{bogus}$ with control flow
    variables $\set{\Rpoisoned,\Rnotpoisoned,\Rneutralized}$.
    \begin{align*}
        P & = U_P  & A & = U_A \\
        \Rnotpoisoned & = \neg(P=1) & 
        \Rpoisoned & = (P=1) \\
        \Rneutralized & = \Rpoisoned \land (A=1) &
        S &= \begin{cases}
            1 & \text{if $\Rnotpoisoned=\top$} \\
            1 & \text{if $\Rneutralized=\top$} \\
            0 & \text{otherwise}
        \end{cases}
    \end{align*}
\end{example}

Here, $P_r$ and $\mathit{NP}_r$ model the control flow after a conditional
checking if the poison was or was not administered. In the positive
branch, $P_r$ is set, and only there $N_r$ can be reached, namely if
the antidote was given and thus the poison neutralized. $S$ is set to
1 once the poison was neutralized or if it was not administered.

This example is known to be
problematic for the counterfactual approach to causation.  For
instance, in Halpern's
modelling~\cite[Example~3.4]{DBLP:conf/ijcai/Halpern15}, the bodyguard
putting in antidote is part of a cause.
Similarly, it is part of a sufficient cause, when intervention is
not restricted, or, if
intervention is restricted, 
the absence of the poison is not part of the sufficient cause
anymore.
Together with Hitchcock, Halpern argues that this cause could
be removed by considering normality
conditions~\cite{DBLP:journals/corr/HalpernH13}.
Blanchard and Schaffer criticise `under-constrained unclarities' of this approach,
calling theorists to `pay more attention to what counts as an apt
causal model [..] before adding more widgets into causal models'~\cite{Blanchard2014CauseWD}.%
\footnote{%
They also provide a more intuitive account of bogus prevention, 
but unfortunately, it only applies to Hitchcocks's definition~\cite{Hitchcock2001-HITTIO},
which has other shortcomings~\cite[Example~4.4]{DBLP:journals/corr/abs-1301-2275}.}.
Putting it simply: it is often unclear what is normal.
Halpern and Hitchcock~\cite[p.~26]{DBLP:journals/corr/HalpernH13} 
provide an ad-hoc solution by adding
a variable representing the chemical reaction neutralizing the poison,
basically introducing the control flow variable $\Rneutralized$ which
is true if the control flow reached a point where the poison was
previously administered ($\Rpoisoned=\top$) and the bodyguard pours
antidote into the coffee.

Intuitively, the antidote should not be considered a cause of the
victim's surviving because it was irrelevant within the \emph{actual}
course of events. We can capture this through the actual value of the
control flow variables. 
DGKSS~\shortcite{Datta2015a} and Beckers~\cite{545701}
require interventions to be consistent with the actual temporal order
in which events occurred. 
Translated to our setting,
this is akin to 
considering $\vec X$ s.t.\ for all $\vec z$
$\sM [\Vres \setminus \vec X\gets \vec z] (C = \vec \top \implies \varphi)$,
where $C$ is the actual control
flow $C:=\set{ V\in \Vrch \mid \sM V=\top}$,
and $\vec \top$ is a sufficiently long sequence consisting only of
$\top$.
But often, the coming about of the actual
course of action is as important as $\varphi$ itself.
\begin{example}[Agreement]\label{ex:agreement}
    $A$ and $B$ vote. If $A=B$, an agreement is reached 
    (signified by $R\in\Vrch$ with $F_R = (A=B)$),
    and the outcome is announced ($O=A$ if $R=1$ and
    otherwise $\bot$). $A$ and $B$ agree on $v$ in actuality.
\end{example}
If the control flow is fixed, then $A=v$ by itself is a sufficient cause of
$O=v$, despite the fact that $A$ and $B$ need to agree to produce any
outcome, as $\vec z$ needs to set $B$ to $v$ in order to preserve
$R=1$.
This illustrates that the actual control flow should not be presumed
a~priori to the cause. On the other hand, if $\varphi$ is established
early on (and monotone in time, e.g., violations of safety properties
like weak secrecy and authentication~\cite{alpern1985defining}),
there is no need to find causes for the control flow after $\varphi$
occurred.
%

Thus we propose the following definition of sufficient cause, which
forbids deviation from the actual control flow and is specific to
causal models with control flow variables.
\begin{definition}[Control flow preserving sufficient cause]\label{def:sufficient-actual-cause} 
    For a causal model $M$ with control flow variables $\Vrch$,
    let $C\subseteq \Vrch$ be the actual control flow
    in context $\vec u$, i.e.,
    the set of control nodes set to $\top$.
    Then, a \emph{control flow preserving sufficient cause} (CFPSC) is defined 
    like a sufficient cause (see
    Definition~\ref{def:sufficient-cause}), but with SF2 modified as
    follows:
    \begin{multline*}
        \textit{CFS2.}~\forall \vec z. 
            (M,\vec u) \vDash \left[ 
            (V^\rch_v)_{v\notin C} \gets \vec \bot,
            \Vres \setminus \vec X\leftarrow \vec z\right] \varphi. \\
    \end{multline*}
\end{definition}

This definition handles 
Example~\ref{ex:agreement}
correctly: 
$(A,B,O)=(v,v,v)$ is the only CFPSC,
capturing the fact that the agreement
needs to be reached in order for $v$ to be announced.
$(A,O)$ by itself is not a CFPSC, witnessed by $\vec z$ setting $B$ to
$v'\neq v$ which results in $R=\bot$.

For \emph{Bogus prevention} (Example~\ref{ex:bogus}),
which 
was the motivation for Halpern and Hitchcock's introduction of
normality conditions
and could previously 
-- in its original formulation~\cite{causal-powers} --
only be treated by means of normality conditions,
our approach
provides  a direct treatment. 
As all
control flow nodes except $C=(\Rnotpoisoned)$ are fixed to zero,
$A$ is not causally relevant. Intuitively, the point at which $A$
matters because the poison was administered ($P_r$) is not available
for any counter factual. The actual control flow $C$, where the
poison is not administered, guarantees 
the victim's surviving ($S=1$),
but needs to be established by \emph{not} administering the poison
($\Rnotpoisoned=\top$). Hence
$(P,S)$ is (the only) CFPSC,  given that $\sM
[(\Rpoisoned,\Rneutralized)\gets (\bot,\bot), A \gets z] S=1$
for $z=1,2$.

\emph{Late preemption} (Example~\ref{ex:ctl-var}
with control flow variables $BS_i$, $i\in\setN_n$)
is also handled correctly;
here $(T_1,\BS)$ is the only CFPSC, i.e., Suzy's throw caused the
bottle to shatter, but not Billie's. Because $C=(\BS_1)$, $\BS_2$ is
set to $\bot$ and Billie's throw has no bearing on the outcome.

\emph{Early preemption}, where the victim is poisoned, but shot before
the poison takes effect (Example~\ref{ex:shot-poisoned},
can be captured similarly, e.g., considering the
model described by Hitchcock~\cite[p. 526]{Hitchcock2007-HITPPA},
or our own adaptation (cf.\ Table~\ref{tab:examples}).
%
Intuitively, $C$ captures the control flow where the victim is shot
but the poison has not yet taken effect. Setting its complement to
$\bot$ correctly disregards the control flow representing the 
poison taking effect due to the victim not being shot.


\section{Related work}\label{sec:related}

We first discuss related work on sufficient causation,
which is the basis for CFPSC, then
related work concerning control flow 
and finally link our insights on control flow to the notion of
structural contingencies in the literature. 

\paragraph{Sufficient causation}

Going back to Lewis~\shortcite{Lewis1973-LEWC},
most definitions of actual causation investigate
claims of the form `had A not occurred, B would not have occurred'.
The notions put forward by Pearl and Halpern~\cite{%
DBLP:conf/ijcai/Halpern15,%
DBLP:journals/corr/abs-1301-2275} follow this idea, which arguably
captures a form of \emph{necessary} causation. We elaborate this point
at the end of this section.

By contrast,
Datta, Garg, Kaynar, Sharma, and Sinha
aim 
at capturing
minimal sequence of
protocol actions \emph{sufficient} to provoke a violation, in order to
provide a tool for forensics as well as a building block for
accountability~\cite{Datta2015a}.
The appeal of sufficient causation is that there is 
a clear interpretation of what it means for $A$ to
be part of a sufficient cause $(A,B)$: $A$ is (jointly with $B$)
causing the event. 
Notions of necessary causation typically lack this kind of
interpretation and require secondary notions like blame to determine
joint responsibility. 
Furthermore,
in particular in the context of distributed systems and program
analysis,
it comes in handy 
for debugging and forensics
that each
sufficient cause 
basically
captures a chain of
events which, on its own, leads to an outcome~\cite{tracing-causes}.

However, necessary causes are more succinct and seem to
capture what is meant by $A$ causes $B$ in natural language better. We
suspect the latter is because natural language often uses
``cause'' to mean ``part of a cause'', but weighing these two
notions against each other is not the scope of this work, and more
ever, depends on the application one has in mind. We are interested in
the coming about of events, so we focus on sufficient causation.

DGKSS's notion of causation~\cite{Datta2015a}
is based on a source code transformation of
non-branching programs within a simple process calculus. We therefore
cast their idea of intervening on events that do \emph{not}
appear in the candidate cause within Pearl's framework (see
Definition~\ref{def:sufficient-cause}).
This methodology allows us
to understand and generalize effects implicit in the
definition of the calculus and translate them back.
Sufficient causes are different from DGKSS's cause traces in that the latter yield
entire traces and that intervention is performed on code instead of
variables. DGKSS's calculus fixes the control flow to the
actual control flow, since their cause traces
contain the line number of the statement effectuating each event.
Every intervention on these cause traces 
(which roughly corresponds to the $z$ in SF2)
needs to contain these 
line numbers in the same order and is disregarded otherwise.
Example~\ref{ex:agreement} was the motivation to depart from this and
capture the coming about of the actual control flow.

Besides DGKSS's work,
there is only little work on sufficient causation
formalizing which
events are jointly sufficient to cause an outcome.
Going back to early attempts of formulating actual causation in a purely
logical framework~\cite{Mackie1965-MACCAC-4},
Halpern~\shortcite{DBLP:conf/kr/Halpern08a} formulates the 
NESS test~\cite{wright1987causation} within  Pearl's causality
framework.
As with the NESS test, this notion only allows for singular causal
judgements (\cite[Theorem~5.3]{DBLP:conf/kr/Halpern08a})
and is thus not suited for capturing joint causation.
%
Halpern~\shortcite{Halpern:2016:AC:3003155} also proposed a notion of
sufficient causes which requires a sufficient cause to be
a) part of all actual causes (hence inducing three variants of the
definition, for each notion of actual causation put forward)
and
b) to ensure $\varphi$ for all
contexts.
If contingencies are restricted to control flow variables and causes
to data flow variables, the first condition is very similar to
CFPSC.
%
By way of the second condition, one avoids that $\varphi$ appears in
all sufficient causes (as it is a trivial actual cause of itself),
however, this condition prohibits capturing joint causation in cases
where rare external circumstances (e.g., strong wind making lighting
the cigarette in Example~\ref{ex:forest-conjunctive} impossible)
could have prevented the outcome altogether, although they actually
did not.
For distributed systems, this is almost always the case due to
possible loss of messages in transition,
hence we consider this criterion too strong
for these purposes.

\paragraph{Causation and control flow}

Instead of control flow,
Beckers extends structured equations with time~\cite[Part II]{545701}. His proposal for actual
causation involves a notion very similar to the NESS criterion. The
notion derived from it does not capture joint causes, but captures many
preemption examples. 

Sharma's thesis extends DGKSS's model with some
control flow: the $\oplus$ operator can capture choices the parties
make, e.g., $V:=0 \oplus 1$ lets the party set $V$ to either 0 or 1~\cite{sharma-thesis}.
As agents can still not branch, the previous discussion on DGKSS's
paper applies.

Besides Beckers and DGKSS, there are other
formalisms for causal models that include
control flow, but do not aim at capturing actual causality~\cite{%
giunchiglia2004nonmonotonic,%
turner1999logic}. 
We purposefully formulated control flow within Pearl's causality
framework, to compare with existing definitions and provide insight
into how control flow can guide the modelling task and improve
definitions.

\paragraph{Relation to structural contingencies in Halpern 2015}

We review Halpern's modification~\shortcite{DBLP:conf/ijcai/Halpern15}
of Halpern and Pearl's definition of actual
causes~\shortcite{DBLP:journals/corr/abs-1301-2275}.
First, because it is well-known,
second, because it employs a secondary notion called \emph{structural
contingencies}, which appears to be related to control flow.
We 
a) give evidence that contingencies relate to control flow
wherever they are successful,
b) show that they are problematic if
data flow is involved, and
c) provide an interpretation of structural contingencies 
in terms of control flow,
supporting the argument that
preemption is first a modelling problem, which should be covered by
distinguishing control flow from data flow, and then a matter of the
definition of a cause.

\begin{definition}[Review: actual cause%
    ]
    \label{def:actual-cause} 
    $\vec X = \vec x$ is an \emph{ actual cause} of $\varphi$ in
    $(M,\vec u)$ if the following three conditions hold.
    \begin{enumerate}
        \item [AC1 and AC3.] Just like SF1 and SF3.
        \item [AC2.] There are
            $\vec W$, $\vec w$ and $\vec x'$ such that 
            $(M,\vec u) \vDash (\vec W = \vec w)$,
            and
                $(M,\vec u) \vDash 
                    [\vec X \leftarrow \vec x', 
                    \vec W \leftarrow \vec w, 
                ] \neg \varphi$.
    \end{enumerate}
\end{definition}

AC2 is a generalisation of the intuition behind Lewis'
counterfactual.
The special case
$\vec W = ()$ corresponds to Lewis' reasoning that
$\vec X = \vec x$ is \emph{necessary} for $\varphi$ to hold, because
there exists a
counterfactual setting $\vec x'$ that negates
$\varphi$.
AC2 weakens this condition in order to capture causes that are
`masked' by other events. 
The set of variables $\vec W$ is called \emph{structural contingency},
as it captures the aspects of the actual situation 
under which $\vec X=\vec x$
is a cause. Variables in $\vec W$ can only be fixed to their actual
values. 

We observe that in all examples in Halpern's paper where 
non-empty contingencies appear
(Suzy-Billy, Ex.~3.6 and Ex.~3.7 after adding variables),
they
exclusively contain variables added to the model to
``describ[e] the mechanism that brings about the result''.
These are the precisely variables we would consider control flow variables.
The only exceptions are Example~3.9b and~3.10, where Halpern points out 
that Definition~\ref{def:actual-cause} gives
unintuitive results. 
Definition~\ref{def:sufficient-actual-cause},
captures Halpern's intuition correctly 
(see Table~\ref{tab:examples}, Train~\cite{Hall2000-HALCAT-7} and Careful
Poisoning).

While structural contingencies work well if they contain control flow
variables, they can give unintuitive results if they 
contain data flow variables.
Consider
the following causal model $M_\mathit{BoS}$ inspired by the
`Battle of Sexes' two-player game.
\begin{example}[Bach or Stravinsky]\label{ex:bos}
    A couple ($P_1,P_2\in\calV$) 
    agreed on meeting this evening,
    but they cannot recall whether they wanted to attend a Bach or
    a Stravinsky concert ($\calR(P_1)=\calR(P_2)=\set{B,S}$). 
    They are taking the train  $(T=1)$, but only if they both go to
    the same concert ($F_C = P_1=P_2$ for the control flow variable $C$).
    Contrary to the awkward situation in the 
    famous two-player game,
    they have left a note $N\in\calU, \calR(N)=\set{B,S}$ 
    in their calendar, which helps them
    remember ($F_{P_1}=N$ and $F_{P_2}=N$).
    Is the fact that the note reminds them to attend the Bach concert ($N=B$)
    a cause for taking the train together ($T=1$)?
\end{example}

Definition~\ref{def:actual-cause},
as well as the definition preceding it~\cite{DBLP:journals/corr/abs-1301-2275},
lead to an unintuitive
result because variables that concern data flow are considered as
contingencies.
No matter which value is chosen for $N$, $T$ is always 1.
 However,
$\vec W$ can be set to $P_1$ or $P_2$,
fixing it to the actual value of $N$.
Hence, for, e.g., $\vec u=(B)$,
it holds that 
$(M_\mathit{BoS},\vec u) \vDash 
[N \gets S, \vec W = (P_2) \gets B ] (T \neq 1)$,
and thus the
choice of the input is considered a cause for $T=1$,
although $1$ is output no matter what the input is.
It appears that $P_1$ and $P_2$ should not be permissible contingencies.
Definition~\ref{def:sufficient-actual-cause} handles this example
correctly: $(P_1,P_2,T)$ is the only CFPSC, as the control flow $(C)$ 
is preserved no matter what value $N$ has, as long as $P_1$ and $P_2$
agree on doing what it says (see Table~\ref{tab:examples}, BoS).
%
%
%
Our conclusion is that 
contingencies should be restricted to control flow variables 
to avoid such spurious causes.

Thus, if the use case allows for making control flow variables
explicit,
we can restrict 
structural contingencies $\vec W$ to control flow variables
and
actual causes $\vec X$ to data flow variables,
and consider AC2' as follows:
\begin{itemize}
    \item [AC2'.] For $\vec X\subseteq \Vres$,
        $C\subseteq \Vrch$ and $\sM \vec C= \vec c$ a subset of
        the actual control flow,
        there is $\vec x'$ such that 
        $\sM [\vec C \gets \vec c, \vec X \gets \vec x'] \neg \varphi$.
\end{itemize}
Under these circumstances, we can relate actual causes and CFPSC
by comparing AC2' to CFS2. A~priori, both are very different:
AC2 (and AC2') formulate
a necessity criterion on $\vec X$,\footnote{%
If $\vec X = \vec x$ is an actual cause under contingency
$\vec W = \vec w$, then $\vec X'= \vec X\cdot\vec W = \vec x\cdot \vec w$ 
is a necessary cause, i.e.,
$(M,u)\models [\vec X' \gets \vec x\cdot \vec w]\neg \varphi$.}
while CFS2 formulates a sufficiency criterion. 
Interestingly, there is a duality between necessary and sufficient
causes~\cite{DBLP:journals/corr/abs-1710-09102}, which we can use to
compare the two w.r.t.\ their treatment of control flow:
If the
set of variables is finite,
the set of all sufficient causes 
can be obtained from the set of all
necessary causes $\calX = \set{\vec X_1,\ldots,\vec X_m}$
by, first, considering them 
as a boolean formula in $\vec X$ in disjunctive
normal form (DNF) ($\vec X$ is in $\calX$ iff either $\vec X$ equals
$\vec{X_1}$, or if it equals $\vec X_2$, etc.), second, computing the
conjunctive normal form of this formula and, finally, switching $\land$ and
$\lor$.
For example, $(A,B)$ is the only sufficient causes in the
conjunctive forest fire example,
and thus
$A$ and $B$ are both necessary causes.
We adapted this result to CFPSCs (see 
the Appendix for theorem and proof)
and obtain a dual definition of
control flow preserving \emph{necessary} causes, where CFS2 translates to
\begin{itemize}
    \item [CFN2.] For $\vec X\subseteq \Vres$ and
        for $C\subseteq \Vrch$ the actual control flow,
        there is $\vec x'$ such that 
        \[ \sM[ (V^\rch_v)_{v\notin C} \gets \vec \bot, \vec X \gets \vec x']\neg\varphi. \]
\end{itemize}
We can now compare AC2' to CFN2 to get a clear picture of
how actual causation handles control flow, if we incorporate
the distinction of control flow variables and data flow variables as
discussed above.
AC2' is much more liberal in how the counter-factual control flow can
be related to the actual control flow.
By choosing an arbitrary subset of all control flow variables, not
only those set to $\top$ or $\bot$,
each counterfactual setting of a control flow variable
may enforce the actual control flow (if it is fixed to $\top$),
prevent counterfactual flow that contradict the actual course of
events (if it is fixed to $\bot$),
but may also just be computed based on the equations (if it is not
part of the subset).
CFN2 is more rigid in this
respect, strictly prohibiting the counter-factual control flow to 
deviate from the actual control flow, but leaving it otherwise free
(motivated by Example~\ref{ex:agreement}).
This explains, e.g., the difference in Weslake's Careful Poisoning
example (see Table~\ref{tab:examples}), where the assassin only adds
poison to the coffee if he is sure the antidote was added previously. 
The antidote is (wrongly to most)
considered an actual cause for the victim surviving,
as the counterfactual where it is not administered
can still consider the control flow where the assassin added the
poison.

We summarize: a restriction of \emph{structural contingencies}
to control flow seems to avoid unintuitive results in some cases,
without losing accuracy on any examples we considered.
Given this restriction, contingencies obtain an interpretation in
terms of the relation between the actual control flow and the
control flow considered in counterfactuals. In comparison to CFPS2,
this relation is much more loose. The Careful Poisoning example
suggests that it is too loose.

\begin{table*}
\begin{minipage}{\textwidth} 
    \resizebox{\textwidth}{!}{
        \raggedright
\begin{tabular}{lp{6.2cm}p{4cm}p{5.2cm}}
{} & {Definition~\ref{def:sufficient-cause}} & {Definition~\ref{def:actual-cause}} & {Definition~\ref{def:sufficient-actual-cause}}\\

\toprule
{\raggedright Forest fire, Ex.~\ref{ex:forest-conjunctive}} & $(\mathit{A}, \mathit{B}, \mathit{FF})$ & $(\mathit{A})$, $(\mathit{B})$, $(\mathit{FF})$ & $(\mathit{A}, \mathit{B}, \mathit{FF})$\\
{\raggedright --- disjunctive (overdet.) \cite[p. 278]{Hall2004-HALTCO-4}} & $(\mathit{B}, \mathit{O})$, $(\mathit{S}, \mathit{O})$ & $(\mathit{B}, \mathit{S})$, $(\mathit{O})$ & $(\mathit{B}, \mathit{O})$, $(\mathit{S}, \mathit{O})$\\
{\raggedright --- disjunctive, ext. \cite[Ex. 3.7]{DBLP:conf/ijcai/Halpern15}} & $(\mathit{MD}, \mathit{L}, \mathit{C}, \mathit{FF})$, $(\mathit{MD}, \mathit{B}, \mathit{C}, \mathit{FF})$, $(\mathit{L}, \mathit{A}, \mathit{C}, \mathit{FF})$ & $(\mathit{MD})$, $(\mathit{L})$, $(\mathit{C})$, $(\mathit{FF})$ & $(\mathit{MD}, \mathit{L}, \mathit{FF})$\\
\midrule{\raggedright Late preemption, Ex.~\ref{ex:ctl-var}} & $(\mathit{T1}, \mathit{BS1}, \mathit{BS})$, $(\mathit{T2}, \mathit{BS1}, \mathit{BS2}, \mathit{BS})$ & $(\mathit{T1})$, $(\mathit{BS1})$, $(\mathit{BS})$ & $(\mathit{T1}, \mathit{BS})$\\
{\raggedright Early preemption \cite[p. 526]{Hitchcock2007-HITPPA}} & $(\mathit{A}, \mathit{D_1}, \mathit{D_2})$, $(\mathit{B}, \mathit{P}, \mathit{D_2})$ & $(\mathit{A})$, $(\mathit{D_1})$, $(\mathit{D_2})$ & $(\mathit{A}, \mathit{D_2})$\\
{\raggedright --- (ctl), Ex.~\ref{ex:shot-poisoned}\footnote{Model with $\Vrch=\set{P_r,S_r,\mathit{PE}_r, \mathit{NS}_r}$ and $P_r=B$, $S_r=A$, $\mathit{NS}_r=\neg A$, $\mathit{PE}_r=P_r \land \neg A$ and $D=1$ iff $S_r=\top$ or $\mathit{PE}_r=\top$.}} & $(\mathit{A}, \mathit{S_r}, \mathit{D})$, $(\mathit{B}, \mathit{P_r}, \mathit{S_r}, \mathit{PE_r}, \mathit{NS_r}, \mathit{D})$ & $(\mathit{A})$, $(\mathit{S_r})$, $(\mathit{D})$ & $(\mathit{A}, \mathit{D})$\\
\midrule{\raggedright Bogus prevention \cite{causal-powers}} & $(\mathit{P}, \mathit{S})$, $(\mathit{A}, \mathit{S})$ & $(\mathit{P}, \mathit{A})$, $(\mathit{S})$ & $(\mathit{P}, \mathit{S})$, $(\mathit{A}, \mathit{S})$\\
{\raggedright --- ad-hoc \cite[p. 29]{DBLP:journals/corr/HalpernH13}} & $(\mathit{P}, \mathit{S})$, $(\mathit{A}, \mathit{S}, \mathit{PN})$ & $(\mathit{P})$, $(\mathit{S})$ & $(\mathit{P}, \mathit{S})$\\
{\raggedright --- (ctl), Ex.~\ref{ex:bogus}} & $(\mathit{P}, \mathit{S}, \mathit{NP_r})$, $(\mathit{A}, \mathit{S}, \mathit{P_r}, \mathit{NP_r}, \mathit{N_r})$ & $(\mathit{P})$, $(\mathit{S})$, $(\mathit{NP_r})$ & $(\mathit{P}, \mathit{S})$\\
{\raggedright --- (ctl, reversed)\footnote{Like Ex. \ref{ex:bogus}, but reversed control flow: $D=0$ iff $\mathit{NP}_r=\top$ (poison not administered) or $\mathit{PN}_r=\top$ (poison neutralised).}} & $(\mathit{P}, \mathit{D}, \mathit{NP_r})$, $(\mathit{A}, \mathit{D}, \mathit{NP_r}, \mathit{P_r}, \mathit{PN_r})$ & $(\mathit{P})$, $(\mathit{D})$, $(\mathit{NP_r})$ & $(\mathit{P}, \mathit{D})$\\
{\raggedright Careful Poisoning \cite[Ex.~11]{weslake2015partial}} & $(\mathit{A}, \mathit{D})$, $(\mathit{P}, \mathit{D})$ & $(\mathit{A})$, $(\mathit{D})$ & $(\mathit{A}, \mathit{D})$, $(\mathit{P}, \mathit{D})$\\
{\raggedright --- (ctl)\footnote{Model with $\Vrch=\set{\mathit{NA}_r,A_r,P_r}$ and $\mathit{NA}_r=\top$ iff $A=0$, $A_r=\top$ iff $A=1$, $P_r=\top$ iff $NA_r=\top\land P=1$ and $D=1$ iff $P_r=\top$.}} & $(\mathit{A}, \mathit{NA_r}, \mathit{P_r}, \mathit{D})$, $(\mathit{NA_r}, \mathit{A_r}, \mathit{P_r}, \mathit{P}, \mathit{D})$ & $(\mathit{A})$, $(\mathit{NA_r})$, $(\mathit{P_r})$, $(\mathit{D})$ & $(\mathit{D})$\\
\midrule{\raggedright Train \cite[Ex.~4]{DBLP:journals/corr/abs-1106-2652}} & $(\mathit{F}, \mathit{RB}, \mathit{A})$, $(\mathit{LB}, \mathit{RB}, \mathit{A})$ & $(\mathit{F}, \mathit{LB})$, $(\mathit{RB})$, $(\mathit{A})$ & $(\mathit{F}, \mathit{RB}, \mathit{A})$, $(\mathit{LB}, \mathit{RB}, \mathit{A})$\\
{\raggedright --- \cite{Hall2000-HALCAT-7}} & $(\mathit{F}, \mathit{RT}, \mathit{A})$ & $(\mathit{F})$, $(\mathit{RT})$, $(\mathit{A})$ & $(\mathit{F}, \mathit{A})$\\
{\raggedright --- (ctl)\footnote{Model with $\Vrch=\set{L_r,R_r}$, where $L_r=\top$ iff $\neg F$, $\mathit{R}_r=\top$ iff $F$ and $A=1$ iff $L_r=\top$ or $R_r=\top$, allowing train to get stuck.}} & $(\mathit{F}, \mathit{R_r}, \mathit{A})$, $(\mathit{L_r}, \mathit{R_r}, \mathit{A})$ & $(\mathit{F})$, $(\mathit{R_r})$, $(\mathit{A})$ & $(\mathit{F}, \mathit{A})$\\
\midrule{\raggedright Prisoner \cite{hopkins2003clarifying}} & $(\mathit{C}, \mathit{D})$ & $(\mathit{C})$, $(\mathit{D})$ & $(\mathit{C}, \mathit{D})$\\
{\raggedright Backup \cite[Ex.~1]{weslake2015partial}} & $(\mathit{T}, \mathit{V})$, $(\mathit{S}, \mathit{V})$ & $(\mathit{T})$, $(\mathit{V})$ & $(\mathit{T}, \mathit{V})$, $(\mathit{S}, \mathit{V})$\\
{\raggedright --- (ctl) \cite[Ex.~1]{weslake2015partial}} & $(\mathit{T}, \mathit{V}, \mathit{T_r})$, $(\mathit{S}, \mathit{V}, \mathit{T_r}, \mathit{NT_r}, \mathit{S_r})$ & $(\mathit{T})$, $(\mathit{V})$, $(\mathit{T_r})$ & $(\mathit{T}, \mathit{V})$\\
{\raggedright Command \cite[Ex.~8]{weslake2015partial}} & $(\mathit{M}, \mathit{C})$ & $(\mathit{M})$, $(\mathit{C})$ & $(\mathit{M}, \mathit{C})$\\
\midrule{\raggedright Agreement, Ex.~\ref{ex:agreement}} & $(\mathit{A}, \mathit{B}, \mathit{R}, \mathit{O})$ & $(\mathit{A})$, $(\mathit{B})$, $(\mathit{R})$, $(\mathit{O})$ & $(\mathit{A}, \mathit{B}, \mathit{O})$\\
{\raggedright BoS, Ex.~\ref{ex:bos}} & $(\mathit{P_1}, \mathit{P_2}, \mathit{C}, \mathit{T})$ & $(\mathit{N})$, $(\mathit{P_1})$, $(\mathit{P_2})$, $(\mathit{C})$, $(\mathit{T})$ & $(\mathit{P_1}, \mathit{P_2}, \mathit{T})$\\
\midrule{\raggedright Switch \cite[p.~16]{weslake2015partial}} & $(\mathit{S}, \mathit{L2}, \mathit{I})$, $(\mathit{L1}, \mathit{L2}, \mathit{I})$ & $(\mathit{S})$, $(\mathit{L2})$, $(\mathit{I})$ & $(\mathit{S}, \mathit{L2}, \mathit{I})$, $(\mathit{L1}, \mathit{L2}, \mathit{I})$\\
{\raggedright Combination Lamp \cite[p.~19]{weslake2015partial}} & $(\mathit{B}, \mathit{C}, \mathit{L})$ & $(\mathit{B})$, $(\mathit{C})$, $(\mathit{L})$ & $(\mathit{B}, \mathit{C}, \mathit{L})$\\
{\raggedright Shock \cite[p.~17]{weslake2015partial}} & $(\mathit{B}, \mathit{C}, \mathit{C1})$ & $(\mathit{A})$, $(\mathit{B})$, $(\mathit{C})$, $(\mathit{C1})$ & $(\mathit{B}, \mathit{C})$\\
{\raggedright Push A \cite[p.~26]{weslake2015partial}} & $(\mathit{P}, \mathit{B}, \mathit{H}, \mathit{D})$, $(\mathit{P}, \mathit{T}, \mathit{H}, \mathit{D})$ & $(\mathit{P})$, $(\mathit{B}, \mathit{T})$, $(\mathit{H})$, $(\mathit{D})$ & $(\mathit{P}, \mathit{B}, \mathit{H}, \mathit{D})$, $(\mathit{P}, \mathit{T}, \mathit{H}, \mathit{D})$\\
{\raggedright --- (ctl)\footnote{Model with $\Vrch=\set{P_r,\mathit{NP}_r}$, where $P_r=\top$ iff $P=1$, $\mathit{NP}_r=\top$ iff $P=0$ and $H=1$ iff either $P_r=\top \land T=1$  or $\mathit{NP}_r=\top \land B=1$.}} & $(\mathit{P}, \mathit{T}, \mathit{P_r}, \mathit{H}, \mathit{D})$ & $(\mathit{P})$, $(\mathit{T})$, $(\mathit{P_r})$, $(\mathit{H})$, $(\mathit{D})$ & $(\mathit{P}, \mathit{T}, \mathit{H}, \mathit{D})$\\
{\raggedright Push B \cite[p.~26]{weslake2015partial}} & $(\mathit{P}, \mathit{T}, \mathit{H}, \mathit{D})$ & $(\mathit{P})$, $(\mathit{T})$, $(\mathit{H})$, $(\mathit{D})$ & $(\mathit{P}, \mathit{T}, \mathit{H}, \mathit{D})$\\
{\raggedright Fancy Lamp \cite[p.~31]{weslake2015partial}} & $(\mathit{A}, \mathit{N3}, \mathit{L})$, $(\mathit{B}, \mathit{N1}, \mathit{L})$ & $(\mathit{A}, \mathit{B})$, $(\mathit{A}, \mathit{N1})$, $(\mathit{B}, \mathit{N3})$, $(\mathit{N1}, \mathit{N3})$, $(\mathit{L})$ & $(\mathit{A}, \mathit{N3}, \mathit{L})$, $(\mathit{B}, \mathit{N1}, \mathit{L})$\\
\midrule{\raggedright Vote \cite[Ex.~4.1]{DBLP:conf/ijcai/Halpern15}} & $(\mathit{V_1}, \mathit{M}, \mathit{P})$, $(\mathit{V_2}, \mathit{M}, \mathit{P})$ & $(\mathit{V_1}, \mathit{V_2})$, $(\mathit{M})$, $(\mathit{P})$ & $(\mathit{V_1}, \mathit{M}, \mathit{P})$, $(\mathit{V_2}, \mathit{M}, \mathit{P})$\\
{\raggedright Ranch \cite[Ex.~3.7]{DBLP:conf/ijcai/Halpern15}} & $(\mathit{A_1}, \mathit{A_2}, \mathit{M_1}, \mathit{O})$ & $(\mathit{A_1})$, $(\mathit{A_2})$, $(\mathit{M_1})$, $(\mathit{O})$ & $(\mathit{A_1}, \mathit{A_2})$\\
{\raggedright Vote 5:2\footnote{Majority vote with 7 participants, 5 of which vote Yea, highlighting the difference between notions based on sufficiency and necessity: 4 Yeas suffice for $O=1$, but if two voters would switch to Nay, the vote would be overturned.}} & $(\mathit{V1}, \mathit{V2}, \mathit{V3}, \mathit{V4}, \mathit{O})$, $(\mathit{V1}, \mathit{V2}, \mathit{V3}, \mathit{V5}, \mathit{O})$, $(\mathit{V1}, \mathit{V2}, \mathit{V4}, \mathit{V5}, \mathit{O})$, $(\mathit{V1}, \mathit{V3}, \mathit{V4}, \mathit{V5}, \mathit{O})$, $(\mathit{V2}, \mathit{V3}, \mathit{V4}, \mathit{V5}, \mathit{O})$ & $(\mathit{V1}, \mathit{V2})$, $(\mathit{V1}, \mathit{V3})$, $(\mathit{V1}, \mathit{V4})$, $(\mathit{V1}, \mathit{V5})$, $(\mathit{V2}, \mathit{V3})$, $(\mathit{V2}, \mathit{V4})$, $(\mathit{V2}, \mathit{V5})$, $(\mathit{V3}, \mathit{V4})$, $(\mathit{V3}, \mathit{V5})$, $(\mathit{V4}, \mathit{V5})$, $(\mathit{O})$ & $(\mathit{V1}, \mathit{V2}, \mathit{V3}, \mathit{V4})$, $(\mathit{V1}, \mathit{V2}, \mathit{V3}, \mathit{V5})$, $(\mathit{V1}, \mathit{V2}, \mathit{V4}, \mathit{V5})$, $(\mathit{V1}, \mathit{V3}, \mathit{V4}, \mathit{V5})$, $(\mathit{V2}, \mathit{V3}, \mathit{V4}, \mathit{V5})$\\
\midrule{\raggedright Pollution, $k=80$ \cite[Ex.~3.11]{DBLP:conf/ijcai/Halpern15}} & $(\mathit{A}, \mathit{D})$ & $(\mathit{A})$, $(\mathit{D})$ & $(\mathit{A}, \mathit{D})$\\
{\raggedright Pollution, $k=50$ \cite[Ex.~3.11]{DBLP:conf/ijcai/Halpern15}} & $(\mathit{A}, \mathit{D})$, $(\mathit{B}, \mathit{D})$ & $(\mathit{A}, \mathit{B})$, $(\mathit{D})$ & $(\mathit{A}, \mathit{D})$, $(\mathit{B}, \mathit{D})$\\
{\raggedright Pollution, $k=120$ \cite[Ex.~3.11]{DBLP:conf/ijcai/Halpern15}} & $(\mathit{A}, \mathit{B}, \mathit{D})$ & $(\mathit{A})$, $(\mathit{B})$, $(\mathit{D})$ & $(\mathit{A}, \mathit{B}, \mathit{D})$\\

\bottomrule\end{tabular}

     }
\end{minipage}
    \caption{Set of causes for various examples from the
    literature. The suffix (ctl) marks models adapted 
    to Definition~\ref{def:ext-causal-model}.}\label{tab:examples}
\end{table*}
\section{Validation}\label{sec:examples}

Our goal was to provide an account of the relation
between control flow and causation;
Definition~\ref{def:sufficient-actual-cause} served this goal by
demonstrating that an explicit treatment of control flow can help (and
is sometimes even necessary) to treat cases where counterfactuals can change the
course of events, e.g., in cases of preemption.
We further validate our definition against a benchmark suite of
\exnumber{} examples from the literature. 
To avoid cherry-picking, we include all examples from~\cite{weslake2015partial} and \cite{DBLP:conf/ijcai/Halpern15}.
The source code is available at \url{https://github.com/rkunnema/causation-benchmark}.
%
We
encourage other researchers to use this suite to test their own
definitions of causality and to extend it with new examples.

We compare Definition~\ref{def:sufficient-actual-cause}
with
Definition~\ref{def:sufficient-cause} for illustration,
and with
Halpern's notion of actual causation 
(see Definition~\ref{def:actual-cause}) as a point of reference.
We omit notions like defaults and normality, as we want to stress that
these are not necessary to deal with examples of preemption,
and favour Halpern's notion over prior variants~\cite{DBLP:journals/corr/abs-1301-2275},
as it distinguishes between joint causation and over-determination.%
\unskip\footnote{%
The Litmus test are the two variants of the
Forest fire example.}

We include the \emph{complete} set of causes as a comma-separated list of sequences,
as, e.g., for Careful Poisoning (ctl) it is important to see that $A$ (the antidote
being administered) is \emph{not} a CFPSC\@. 
Many examples in the literature do not capture control flow in their
modelling, in which case we present results for the original
modelling and a modelling according to Definition~\ref{def:ext-causal-model}.
In all examples we get results that are satisfying according to the discussion
in the literature we cite.\footnote{For Switch and Fancy Lamp, our
Definition~\ref{def:sufficient-actual-cause} captures that
$S$, respectively, $A$ are not necessary for the
outcome by giving a sufficient cause which does
not include them.} 
For lack of space we refer to the cited
literature for deeper explanation.
The only examples we added (not counting adaptations to Definition~\ref{def:ext-causal-model})
are Examples~\ref{ex:agreement} and~\ref{ex:bos}, which we explained already, and Vote 5:2,
which contrasts the differing views sufficient causes and actual/necessary causes  
have to offer.


\section{Conclusion and future work}

In cases where control flow can be made explicit,
it is worth doing so, as this helps to deal with problematic cases like
preemption. We discussed in what way it should be taken into account, and
proposed a blue-print for doing so, as well as a definition of causality
that preserves the actual course of events.
This definition is simple and intuitive,
does not require secondary notions,
captures joint causation
and
handles all \exnumber{} example in our benchmark correctly
(with respect to the respective author's views).
Such a definition is useful for computer programs and distributed systems,
as the temporal order of events between communicating agents
can be captured by a non-deterministic scheduler simulating them.
A translation from Petri nets, Kripke structures or process calculi to extended
causal models that adheres to the modelling principles discussed in this work
can be used to argue the soundness of causality notions formulated within
these formalisms.
In fact, we advocate this approach, as we believe that a thorough discussion of
causality requires a common language.

Vice versa, the causality literature stands to benefit from such
translations. Due to the generality of Pearl's framework, modelling is
more art than science, which raises concerns about the falsifiability
of theories of causation.
For bogus prevention and other
difficult examples, e.g., late preemption,
the `correct' modelling has been
debated again and again for each use case specifically.
By transferring and analysing  existing domain-specific definitions, e.g., DGKSS's definition,
to Pearl's framework,
we can find a common ground for the discussion of the `right'
modelling, and encourage a similar treatment for other domains where
causality is of interest.
This way, experiences and observation in well-understood domains can be fed
back to the general case and modelling principles be exposed that are otherwise
easily overlooked when modelling abstract scenarios ad hoc.


\appendix
\bibliographystyle{eptcs}
\bibliography{references}
\section{Appendix: Duality between CFS2 and CFN2}

Fix some finite set $\Vres\subset \calV$ and some ordering
$\set{V_1,\ldots,V_n} = \Vres$ and 
let $\overline X$ denote the following representation of 
$X \subseteq \Vres$ relative to $\Vres$:
$\overline X \defeq{} (1_X(V_1),\ldots,1_X(V_n))$.
Any set of sets of variables $\calX=\vec X_1,\ldots, \vec X_m$ can be
represented as a boolean formula in disjunctive normal form (DNF) that
is true whenever $\overline X$ is the bitstring representation of
$X\subseteq\Vres$ such that $X\in \calX$:
$
(\overline X=\overline X_1 \lor \overline X=\overline X_2 \lor \cdots \lor \overline X= \overline X_m)$.
Here, $\overline X=\overline X_i$ 
merely compares two bitstrings for equality, i.e., 
it is a conjunction 
$\bigwedge_{j\in \setN_{n}} \overline X|_j = \overline X_i|_j$.

\begin{theorem}[control flow preserving necessary 
    causes]\label{thm:sufficient-necessary}
    For $\calX$ the set of (not necessarily minimal) CFPSCs, i.e.,
    adhering to SF1 and CFS2,
    let $\overline \calX$ be the DNF representation of $\calX$.
    Then the set of 
    (not necessarily minimal)
    control flow preserving actual causes, i.e., adhering to AC1 and CFN2,
    is represented by
    $\overline \calY$, which is obtained from $\overline \calX$
    by transforming $\overline \calX$ into CNF and switching $\lor$
    and $\land$. The same holds for the other direction.
\end{theorem}
\begin{proof}
    Fix an arbitrary model $M$, context $\vec u$ and let
    $C\subseteq \Vrch$
    be the actual control flow in this context.
    By Definition~\ref{def:sufficient-actual-cause},
    CFS2, we can rephrase the assumption $\calX$ as follows: 
    For all sequences of variables
    $\vec X$,
    \begin{multline}\label{eq:start}
     \vec X\in\calX \iff 
        \forall \vec z. 
            (M,\vec u) \vDash \left[ (V^\rch_v)_{v\notin C} \gets \vec \bot,
            \Vres \setminus \vec X\leftarrow \vec z\right] \varphi.
    \end{multline}
    As $\calX$ is finite, i.e., $\calX=\set{\vec X_1,\ldots, \vec
    X_n}$, we can write $\vec X\in\calX$
    as a boolean function over $\set{0,1}^n$:
    \[
    \vec X\in \calX
    \iff
    (\overline X = \overline X_1) \lor \cdots \lor (\overline X= \overline X_m).
    \]
    Like any boolean function, this function
    can be transformed into canonical CNF and thus 
    the right-hand side can be expressed as
    $c_1 \land \cdots \land c_{k}$ with some conjuncts 
    $c_i$ of form $\bigvee_{j\in\setN_{n}} (\neg) \overline X|_j$.
    Whatever these conjuncts are,
    we insert this CNF on the left-hand side of \eqref{eq:start} and
    negate both sides. Hence, for all $\vec X$,
    \begin{multline*}
        \neg c_1 \lor \cdots \lor \neg c_{k} 
        \iff 
        \exists \vec z. 
            (M,\vec u) \vDash \left[ (V^\rch_v)_{v\notin C} \gets \vec \bot,
            \Vres \setminus \vec X\leftarrow \vec z\right] \neg \varphi. 
    \end{multline*}
    (Note that $c_1,\ldots, c_k$ depend on $\vec X$.)
    We rename $\vec X$ to $\vec Z$ and $\vec z$ to $\vec x'$.
    Let $\{^b / _a\}$ denote
    $b$ literally replacing $a$. Thus, for all $\vec Z$,
\begin{multline*}
        \neg c_1\left\{^{\overline Z}/_{\overline X}\right\} \lor \cdots \lor \neg c_{k}\left\{^{\overline Z}/_{\overline X}\right\} 
        \iff 
        \exists \vec x'. 
            (M,\vec u) \vDash \left[ (V^\rch_v)_{v\notin C} \gets \vec \bot,
            \Vres \setminus \vec Z\leftarrow \vec x'\right] \neg\varphi. 
    \end{multline*}
        We can replace $\Vres\setminus \vec Z$ by a new variable $\vec X$
        and quantify over $\vec X$ again. This is valid, as
        $\vec X \mapsto \Vres\setminus \vec X$
        is a bijection between the domain of $\vec X$ and the domain
        of $\vec Z$. Note that 
        $\Vres \setminus (\Vres \setminus \vec X) = \vec X$
        and
        $\overline Z = \overline{(\Vres \setminus \vec X)} = \neg
        \overline X$.
        Thus for all $\vec X$
\begin{multline*}
        \neg c_1\left\{^{\neg \overline X}/_{\overline X}\right\} \lor \cdots \lor \neg c_{k}\left\{^{\neg \overline X}/_{\overline X}\right\} 
        \iff
        \exists \vec x'. 
            \sM
            \left[ (V^\rch_v)_{v\notin C} \gets \vec \bot,
            \vec X
            \leftarrow \vec x'\right] \neg\varphi. 
    \end{multline*}
    As each conjunct $c_i$ is a disjunction, the negation of $c_i$
    with $\overline X$ substituted by $\neg \overline X$ can be
    obtained by switching $\lor$ and $\land$. The resulting term is,
    again, a boolean formula in DNF, so $\overline \calX$ transforms
    into $\calX$ easily. 
    The reverse direction proceeds with the exact same steps, modulo
    variables naming.
\end{proof}

\end{document}